\documentclass[submission,copyright,creativecommons]{eptcs}
\providecommand{\event}{FROM 2019} 
\usepackage{breakurl}             
\usepackage{underscore}           
\usepackage{graphicx}
\usepackage{amssymb}
\usepackage{amsthm}
\usepackage{stmaryrd}

\newcommand{\svdash}{\vdash_s}
\newcommand{\PHbH}{{HbH}}
\newcommand{\inn}{{\it in}}
\newcommand{\out}{{\it out}}
\newcommand{\pt}{{\it pt}}

\newcommand{\tst}{\overline{\overline{t}}}

\newcommand{\asspt}{{\overleftarrow{\pt}}}
\newcommand{\Tst}{\overline{\overline{T}}}

\newcommand{\Asspt}{{\overleftarrow{A_{\pt}}}}
\newenvironment{todo}{\bigskip\hrule\medskip\noindent}{\medskip\hrule\bigskip}

\newtheorem{theorem}{Theorem}
\newtheorem{definition}{Definition}
\newtheorem{lemma}{Lemma}
\newtheorem{remark}{Remark}

\title{Explaining SDN Failures via Axiomatisations}
\author{Georgiana Caltais
\thanks{This work was supported by the DFG project ``CRENKAT'', proj. no. $398056821$.}
\institute{University of Konstanz, Germany}
\email{Georgiana.Caltais@uni-konstanz.de}
}
\def\titlerunning{Explaining SDN Failures}
\def\authorrunning{Georgiana Caltais}
\begin{document}
\maketitle

\begin{abstract}
This work introduces a concept of explanations with respect to the violation of safe behaviours within software defined networks (SDNs) expressible in NetKAT. The latter is a network programming language that is based on a well-studied mathematical structure, namely, Kleene Algebra with Tests (KAT). Amongst others, the mathematical foundation of NetKAT gave rise to a sound and complete equational theory. In our setting, a safe behaviour is characterised by a NetKAT policy which does not enable forwarding packets from ingress to an undesirable egress. Explanations for safety violations are derived in an equational fashion, based on a modification of the existing NetKAT axiomatisation.
\end{abstract}

\section{Introduction}
\label{sec::intro}

Explaining systems failure has been a topic of interest for many years now.
Techniques such as Fault tree analysis (FTA) and Failure mode and effects analysis (FMEA)~\cite{Buckl:2007:MAD:1927558.1927572}, for instance, have been proposed and widely used by reliability engineers in order to understand how systems can fail, and for debugging purposes.

In the philosophy of science there is a considerable amount of research on counterfactual causal reasoning which, ultimately, can be exploited in explaining system failures as well.
A notion of causality that is frequently used in relation to technical systems 
relies on counterfactual reasoning. The results in~\cite{Lew73} formulate the counterfactual
argument, which defines when an event is considered a cause for some
effect (or hazardous situation) in the following way: a) whenever the event
presumed to be a cause occurs, the effect occurs as well, and b) when the presumed cause does not occur,
the effect will not occur either. Nevertheless, this formulation of causality is considered too simple for explaining complex logical relationships leading to undesired situations. Consequently, the work in~\cite{DBLP:conf/sum/Halpern11} introduces a notion of complex logical events based on boolean equation systems and proposes a number of conditions under which an event can be considered causal for an effect.

The seminal work in~\cite{DBLP:conf/sum/Halpern11} has been adopted in various settings.
Closely related to the explanation of failures based on adoptions of~\cite{DBLP:conf/sum/Halpern11} are the results in~\cite{DBLP:journals/fmsd/BeerBCOT12,DBLP:conf/vmcai/Leitner-FischerL13,DBLP:journals/corr/CaltaisLM16}, for instance. All these works adjust the definition of causality in~\cite{DBLP:conf/sum/Halpern11} to the setting of system executions leading to a failure.
The approach in~\cite{DBLP:journals/fmsd/BeerBCOT12} considers one such execution at a time, and uses counterfactual causal reasoning for identifying the points in the trace that are relevant for the failure. The results in~\cite{DBLP:conf/vmcai/Leitner-FischerL13,DBLP:journals/corr/CaltaisLM16}, on the other hand, aim at discovering the causal explanations for all failures in a system, and strongly rely on model-checking based techniques~\cite{DBLP:books/daglib/0020348}.

In this paper we focus on explaining violations of safe behaviours in software defined networks (SDNs).
{Software defined networking} is an emerging approach to network programming in a setting where the network control is decoupled from the forwarding functions. This makes the network control directly programable, and more flexible to change.
SDN proposes open standards such as the OpenFlow~\cite{DBLP:journals/ccr/McKeownABPPRST08} API defining, for instance, low-level languages for handling switch configurations. Typically, this kind of hardware-oriented APIs are not intuitive to use in the development of programs for SDN platforms.
Hence, a suite of network programming languages raising the level of abstraction of programs, and corresponding verification tools have been recently proposed~\cite{DBLP:conf/icfp/FosterHFMRSW11,DBLP:conf/dsl/VoellmyH09,DBLP:conf/sigcomm/VoellmyWYFH13}.
It is a known fact that formal foundations can play an important role in guiding the development of programming languages and associated verification tools, in accordance with an intended semantics obeying essential (behavioural) laws.
Correspondingly, the current paper is targeting {NetKAT}~\cite{DBLP:conf/popl/AndersonFGJKSW14,DBLP:conf/popl/FosterKM0T15}  --a formal framework for specifying and reasoning about networks, integrated within the Frenetic suite of network management tools~\cite{DBLP:conf/icfp/FosterHFMRSW11}. More precisely, we will exploit the sound and complete axiomatisation of NetKAT in~\cite{DBLP:conf/popl/AndersonFGJKSW14} in order to derive explanations of safety failures in a purely equational fashion. It is well known that equational reasoning could alleviate the state explosion issue characteristic to model-checking, by equating terms equivalent modulo associativity, commutativity and indepotnecy, for instance.

Related to the current work,
the results in~\cite{root-fail-netkat} introduce a framework for automated failure localisation in NetKAT.
The approach in~\cite{root-fail-netkat} relies on the generation of test cases based on the network specification, further used to monitor the network traffic accordingly and localise faults whenever tests are not satisfied.
In contrast, our approach falls under the umbrella of equivalence checking, and it provides an explanation for all possible failures, irrespective of particular input packets.

\bigskip
\noindent
{\it Our contributions.} In this paper we introduce a concept of \emph{safety in NetKAT} which, in short, refers to the impossibility of packets to travel from a given ingress to a specified hazardous egress. Then, we propose a notion of \emph{safety failure explanation} which, intuitively, represents the set of finite paths within the network, leading to the hazardous egress. Eventually, we provide a modified version of the original axiomatisation of NetKAT exploited in order to \emph{automatically compute the safety failure explanations}, if any. The axiomatisation employs a proposed \emph{star-elimination construction} which enables the sound extraction of explanations from Kleene $*$-free NetKAT programs. 

\bigskip
\noindent
{\it Structure of the paper.} In Section~\ref{sec::prelim} we provide an overview of NetKAT and the associated sound and complete axiomatisation. In Section~\ref{sec::safety} we define the concept of safety in NetKAT.
In Section~\ref{sec::expln-fail} we introduce the notion of safety failure explanation and the axiomatisation which can be exploited in order to compute such explanations. In Section~\ref{sec::discussion} we draw the conclusions and pointers to future work.

\section{Preliminaries}
\label{sec::prelim}


As pointed out in~\cite{DBLP:conf/popl/AndersonFGJKSW14},
a network can be interpreted as an
automaton that forwards packets from one node to another along the links
in its topology. This lead to the idea of using regular expressions --the language of finite automata--, for expressing networks. A path is encoded as a concatenation of processing
steps ($p.q.\ldots$), a set of paths is encoded as a union of paths
($p+q+\ldots$) whereas iterated processing is encoded using Kleene $*$.
This paves the way to reasoning about properties of networks using Kleene Algebra with Tests (KAT)~\cite{DBLP:journals/toplas/Kozen97}. KAT incorporates both Kleene Algebra~\cite{DBLP:journals/iandc/Kozen94} for
reasoning about network structure and Boolean Algebra for reasoning
about the predicates that define switch behaviour.

\begin{figure}[ht]
\center
\includegraphics[scale=0.18,origin=c]{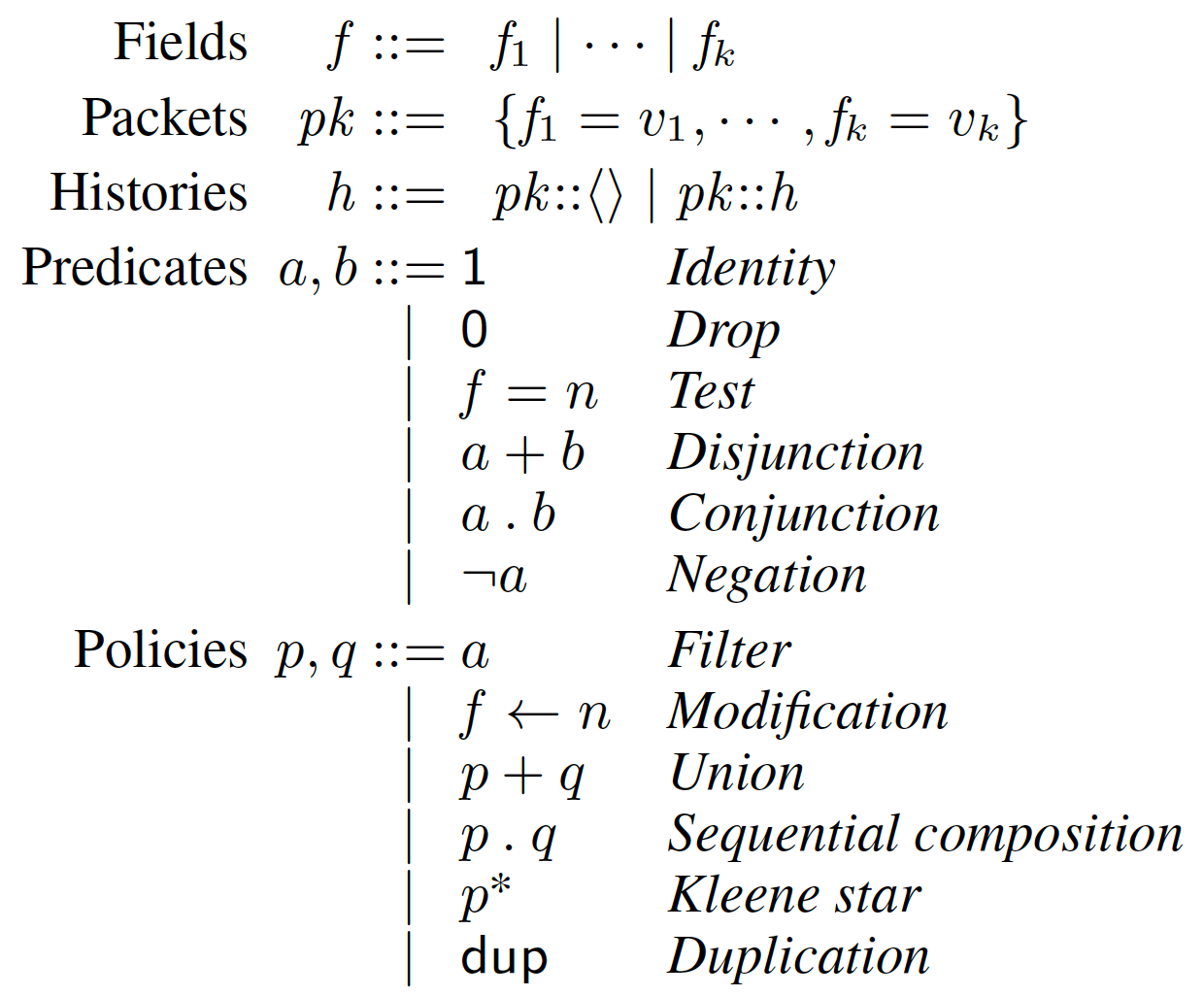}
\caption{NetKAT syntax~\cite{DBLP:conf/popl/AndersonFGJKSW14} \label{fig::syntax}} 
\end{figure}

\begin{figure}[ht]
\center
\includegraphics[scale=0.18,origin=c]{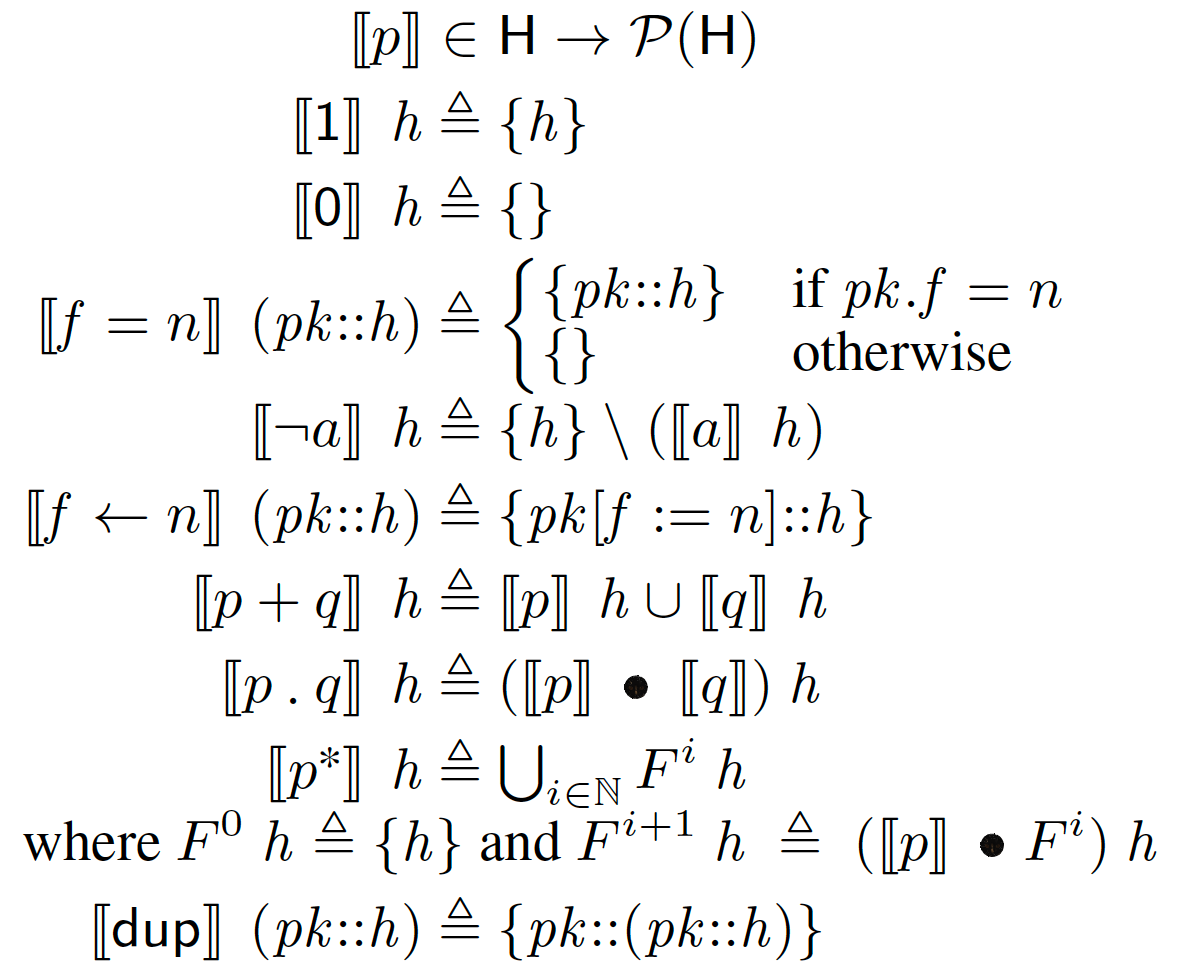}
\caption{NetKAT semantics~\cite{DBLP:conf/popl/AndersonFGJKSW14} \label{fig::semantics}} 
\end{figure}

\begin{figure}[ht]
\center
\[
\begin{array}{r@{}c@{}ll|r@{}c@{}ll}
p + (q+r) & \,\equiv\, & (p+q) + r & \textnormal{KA-PLUS-ASSOC}~ &
~a + (b.c) & \,\equiv\, & (a+b).(a+c) & \textnormal{BA-PLUS-DIST}\\
p + q& \equiv & q+p & \textnormal{KA-PLUS-COMM} &
a +1 & \equiv & 1 & \textnormal{BA-PLUS-ONE}\\
p + 0& \equiv &p & \textnormal{KA-PLUS-ZERO} &
a +\neg a & \equiv & 1 & \textnormal{BA-EXCL-MID}\\
p + p & \equiv & p & \textnormal{KA-PLUS-IDEM} &
a .b & \equiv & b.a & \textnormal{BA-SEQ-COMM}\\
p . (q.r) & \,\equiv\, & (p.q) . r & \textnormal{KA-SEQ-ASSOC}~ &
a . \neg a & \,\equiv\, & 0 & \textnormal{BA-CONTRA}\\
1.p & \,\equiv\, & p & \textnormal{KA-ONE-SEQ}~ &
a . a & \,\equiv\, & a & \textnormal{BA-SEQ-IDEM}\\
p.1 & \,\equiv\, & p & \textnormal{KA-SEQ-ONE}~ & & &  & \\
p.(q+r) & \,\equiv\, & p.q + p.r & \textnormal{KA-SEQ-DIST-L}~ & & &  & \\
(p+q).r & \,\equiv\, & p.r + q.r & \textnormal{KA-SEQ-DIST-R}~ & & &  & \\
0.p & \,\equiv\, & 0 & \textnormal{KA-ZERO-SEQ}~ & & &  & \\
p.0 & \,\equiv\, & 0 & \textnormal{KA-ZERO-SEQ}~ & & &  & \\
1 + p.p^* & \,\equiv\, & p^* & \textnormal{KA-UNROLL-L}~ & & &  & \\
1 + p^*.p & \,\equiv\, & p^* & \textnormal{KA-UNROLL-R}~ & & &  & \\
q+p.r \leq r  & \,\Rightarrow\, & p^* .q\leq r & \textnormal{KA-LFP-L}~ & & &  & \\
p+q.r \leq q  & \,\Rightarrow\, & p.r^*\leq q & \textnormal{KA-LFP-R}~ & & &  & \\
\end{array}
\]
\caption{Kleene Algebra Axioms \& Boolean Algebra Axioms~\cite{DBLP:conf/popl/AndersonFGJKSW14} \label{fig::axiom-KAT}} 
\end{figure}

\begin{figure}[ht]
\center
\[
\begin{array}{r@{}c@{}ll}
f \leftarrow n . f' \leftarrow n' & \,\equiv\, & f' \leftarrow n' . f \leftarrow n, \textnormal{if } f \not = f' ~~& \textnormal{PA-MOD-MOD-COMM} \\
f \leftarrow n . f' = n' & \equiv & f' \leftarrow n' . f = n, \textnormal{if } f \not = f' & \textnormal{PA-MOD-FILTER-COMM} \\
{\bf dup} . f = n & \equiv & f = n . {\bf dup} & \textnormal{PA-DUP-FILTER-COMM} \\
f \leftarrow n . f = n & \equiv & f \leftarrow n & \textnormal{PA-MOD-FILTER}\\
f = n . f \leftarrow n & \equiv & f = n & \textnormal{PA-FILTER-MOD}\\
f \leftarrow n . f \leftarrow n' & \equiv & f \leftarrow n' & \textnormal{PA-MOD-MOD}\\
f = n . f = n' & \equiv & 0, \textnormal{if } n \not= n'& \textnormal{PA-CONTRA}\\
\Sigma_i f = i & \equiv & 1 & \textnormal{PA-MATCH-ALL}
\end{array}
\]
\caption{Packet Algebra Axioms~\cite{DBLP:conf/popl/AndersonFGJKSW14} \label{fig::axiom-packet}} 
\end{figure}

NetKAT \emph{packets} ${\it pk}$ are encoded as sets of fields $f_i$ and associated values $v_i$ as in Figure~\ref{fig::syntax}. \emph{Histories} are defined as lists of packets, and are exploited in order to define the semantics of NetKAT policies/programs as in Figure~\ref{fig::semantics}. NetKAT \emph{policies} are recursively defined as: {predicates}, field modifications $f \leftarrow n$, union of policies $p+q$ ($+$ plays the role of a multi-casting like operator), sequencing of policies $p.q$, repeated application of policies $p^*$ (the Kleene $*$) and duplication $\bf dup$ (that saves the current packet at the beginning of the history list). \emph{Predicates}, on the other hand, can be seen as filters. The constant predicate $0$ drops all the packets, whereas its counterpart predicate $1$ retains all the packets. The test predicate $f = n$ drops all the packets whose field $f$ is not assigned value $n$. Moreover, $\neg a$ stands for the negation of predicate $a$, $a+b$ represents the disjunction of predicates $a$ and $b$, whereas $a.b$ denotes their conjunction.

Let $H$ be the set of all histories, and ${\cal P}(H)$ be the powerset of $H$.
In Figure~\ref{fig::semantics},
the semantic definition of a NetKAT policy $p$ is given as a function $\llbracket p\rrbracket$ that takes a history $h \in H$ and produces a (possibly empty) set of histories in ${\cal P}(H)$.
Some intuition on the semantics of policies was already provided in the paragraph above.
In addition, note that negated predicates drop the packets not satisfying that predicate: $\llbracket \neg a\rrbracket h = \{h\} \setminus \llbracket a \rrbracket h$. The sequential composition of policies $\llbracket p.q\rrbracket$ denotes the Kleisli composition $\bullet$ of the functions  $\llbracket p \rrbracket$ and  $\llbracket q\rrbracket$ defined as:
\[
(f \bullet g) x \triangleq \bigcup \{g\, y \mid y \in f\, x\}.
\]
The repeated iteration of policies is interpreted as the union of $F^i\,h$, where the semantics of each $F^i $ coincides with the semantics of the policy resulted by concatenating $p$ with itself for $i$ times, for $i \in {\mathbb N}$.

In Figure~\ref{fig::axiom-KAT} and Figure~\ref{fig::axiom-packet} we recall the sound and complete axiomatisation of NetKAT. The Kleene Algebra with Tests axioms in Figure~\ref{fig::axiom-KAT}, have been formerly introduced in~\cite{DBLP:journals/toplas/Kozen97}.
Completeness of NetKAT results from the packet algebra axioms in Figure~\ref{fig::axiom-packet}.
The axiom PA-MOD-MOD-COMM stands for the commutativity of different field assignments, whereas PA-MOD-FILTER-COMM denotes the commutativity of different field assignments and tests, for instance.
PA-MOD-MOD states that two subsequent modifications of the same field can be reduced to capture the last modification only. The axiom PA-CONTRA states that the same field of a packet cannot have two different values, etc.

We write $\vdash e \equiv e'$, or simply $e \equiv e'$, whenever the equation $e \equiv e'$ can be proven according to the NetKAT axiomatisation.

Assume, for an example, a simple network consisting four hosts $H_1, H_2, H_3$ and $H_4$ communicating with each other via two switches $A$ and $B$, via the uniquely-labeled ports $1, 2, \ldots, 6$, as illustrated in Figure~\ref{fig::simpl-net}. The network topology can be given by the NetKAT expression:
\begin{equation}\label{eq::topology}
\begin{array}{rcl}
t & \triangleq & \pt= 5 . \pt \leftarrow 6 + \pt = 6 . \pt \leftarrow 5 \,+\\
& & \pt=1 + \pt=2 + \pt=3 + \pt=4
\end{array}
\end{equation}
For an intuition, in~(\ref{eq::topology}), the expression $ \pt= 5 . \pt \leftarrow 6 + \pt = 6 . \pt \leftarrow 5$ encodes the internal link $5 - 6$ by using the sequential composition of a filter that keeps the packets at one end of the link and a modification that updates the $\pt$ fields to the location at the other end of the link.
A link at the perimeter of the network is encoded as a filter that returns the packets located at the ingress port.

\begin{figure}
\center
\includegraphics[scale=0.3,origin=c]{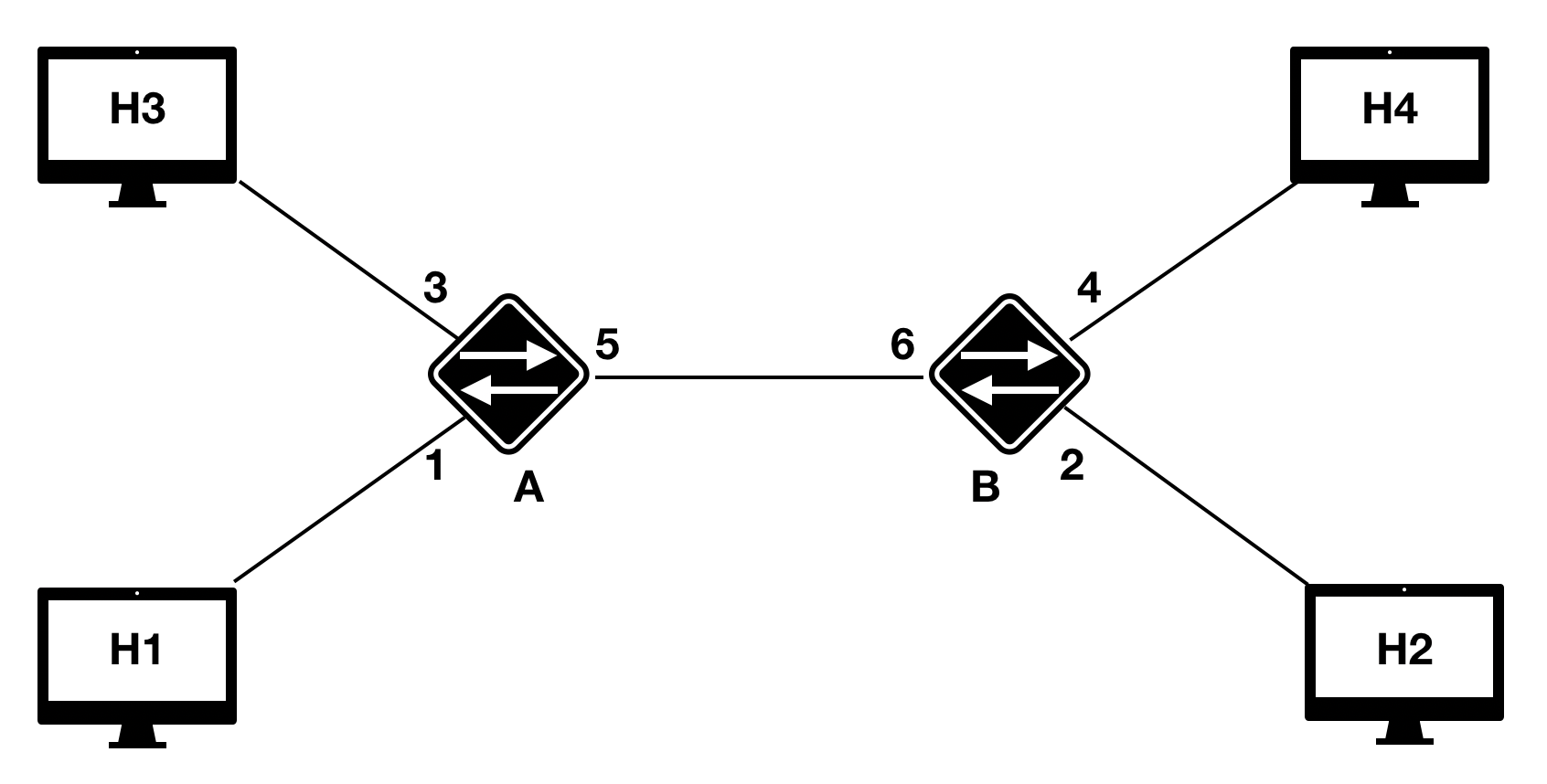}
\caption{A simple network \label{fig::simpl-net}} 
\end{figure}

Furthermore, assume a programmer $P_1$ as in~\cite{DBLP:conf/popl/AndersonFGJKSW14} which has to encode a switch policy that only enables transferring packets from $H_1$ to $H_2$. $P_1$ might define the ``hop-by-hop'' policy in~(\ref{eq::pol-p1}), where each summand stands for the forwarding policy on switch $A$ and $B$, respectively. 
\begin{equation}\label{eq::pol-p1}
p_1  \triangleq \pt = 1 . \pt \leftarrow 5 + \pt = 6 . \pt \leftarrow 2
\end{equation}
In the expression above, the NetKAT expression $\pt = 1 . \pt \leftarrow 5$ sends the packets arriving at port $1$ on switch $A$, to port $5$, whereas $\pt = 6 . \pt \leftarrow 2$ sends the packets at port $6$ on switch $B$, to port $2$. 

At this point, from $P_1$'s perspective, the end-to-end behaviour of the network is defined as:
\begin{equation}\label{eq::p1-e-to-e}
(\pt = 1) . (p_1 . t)^* . (\pt = 2)
\end{equation}
In words: packets situated at ingress port $1$ (encoded as $\pt = 1$) are forwarded to egress port $2$ (encoded as $\pt = 2$) according to the switch policy $p_1$ and topology $t$ (encoded as $(p_1 . t)^*$).

More generally, assuming a switch policy $p$, topology $t$, ingress ${\it in}$ and egress ${\it out}$, the \emph{end-to-end behaviour} of a network is defined as:
\begin{equation}\label{eq::net-beh}
{\it in} . (p . t)^* . {\it out}
\end{equation}

Hence, based on~(\ref{eq::p1-e-to-e}), in order to asses the correctness of $P_1$'s program, one has to show that:
\begin{enumerate}
\item packets at port $1$ reach port $2$, i.e.,
\begin{equation}\label{eq::p1-reach}
\vdash (\pt = 1) . (p_1 . t)^* . (\pt = 2) \not \equiv 0
\end{equation}
\item no packets at port $1$ can reach ports $3$ or $4$, i.e.,
\begin{equation}\label{eq::p1-not-reach}
\vdash (\pt = 1) . (p_1 . t)^* . (\pt = 3 + \pt = 4) \equiv 0.
\end{equation}
\end{enumerate}

By applying the NetKAT axiomatisation, the inequality in~(\ref{eq::p1-reach}) can be equivalently rewritten as:
\begin{equation}\label{eq::equiv-one-reach-two}
\vdash \pt = 1 . \pt \leftarrow 2 + e\not \equiv 0
\end{equation}
with $e$ a NetKAT expression. Observe that $\pt = 1 . \pt \leftarrow 2 $ cannot be reduced further. Hence, the inequality in~(\ref{eq::p1-reach}) holds, as $\pt = 1 . \pt \leftarrow 2 \not \equiv 0$. In other words, the packets located at port $1$ reach port $2$.
Showing that no packets at port $1$ can reach port $3$ or $4$ follows in a similar fashion.

\section{Safety in NetKAT}
\label{sec::safety}

As discussed in the previous section, arguing on equivalence of NetKAT programs can be easily performed in an equational fashion. One interesting way of further exploiting the NetKAT framework is to formalise and reason about well-known notions of program correctness such as safety, for instance. Intuitively, a safety property states that ``something bad never happens''. Ideally, the framework would provide a positive answer whenever a certain safety property is satisfied by the program, and an explanation of what went wrong in case the property is violated.

Consider the example of programmer $P_1$. The ``bad'' thing that could happen is that his switch policy enabled packets to reach ports $3$ or $4$. One can encode such a hazard via the egress policy ${\it out} \triangleq \pt = 3 + \pt = 4$, and the whole safety requirement as in~(\ref{eq::p1-not-reach}). As previously discussed, the NetKAT axiomatisation provides a positive answer with respect to the satisfiability of the safety requirement in~(\ref{eq::p1-not-reach}).

We further proceed by formalising a notion of \emph{port-based hop-by-hop} policy and a \emph{safety} concept in NetKAT.
Let $f_1, \ldots, f_n$ be the list of fields defining a packet, including the port field $\pt$.
Moreover, for simplicity, and without loss of generality, assume no two different ports have the same value/identifier in the network.
Let $\Asspt$ represent the set of all possible port assignments $\pt \leftarrow v$ for some value $v$.
Let $\Tst$ denote the set of all possible tests of the form $f = v$ for some field $f \in \{f_1, \ldots, f_n\}$ and value $v$. 
We write:
\begin{itemize}
 \item $\Tst^*$ to represent the set of all sequences of tests in $\Tst$
 \item $\Sigma_{\Tst^*}$ for the set of all tests $\tst = \Sigma_{i = \{i, \ldots, m\}} \tst^*_i$ with $\tst^*_i \in \Tst^*$, for all $i \in \{1, \ldots m\}$
 \end{itemize}

%

\begin{definition}[{\PHbH} Switch Policy]\label{def::hbh}
A switch policy $p$ is called \emph{port-based hop-by-hop ({\PHbH})} if it is defined as:
\begin{equation}\label{eq::HbH}
p \triangleq \Sigma_{i \in \{1, \ldots, m\}} \, \tst^*_i \,.\, \asspt_i
\end{equation}
where $\tst^*_i \in {\Tst^*}$ and $\asspt_i \in \Asspt$, for all $i \in \{1, \ldots, m\}$.

We call $m$ the \emph{size} of the {\PHbH} policy $p$.
\end{definition}

\begin{definition}[In-Out Safe]\label{def::safety}
Assume a network topology $t$, a {\PHbH} switch policy $p$, an ingress policy $\inn \in \Sigma_{\Tst^*}$, and an egress policy $\out \in \Sigma_{\Tst^*}$, the latter encoding the hazard, or the ``bad thing''. The end-to-end network behaviour is \emph{in-out safe} whenever the following holds:
\begin{equation}
\vdash \inn . (p . t)^* . \out \equiv 0.
\end{equation}
\end{definition}
Intuitively, none of the packages at ingress $\inn$ can reach the ``hazardous'' egress $\out$ whenever forwarded according to the switch policy $p$, across the topology $t$.

\begin{remark}
A notion of reachability within NetKAT-definable networks was proposed in~\cite{DBLP:conf/popl/AndersonFGJKSW14} based on the existence of a non-empty packet history that, in essence, records all the packet modifications produced by the policy $ \inn . (p . t)^* . \out $. This is more like a model-checking-based technique that enables identifying \emph{one} counterexample witnessing the violation of the property $ \inn . (p . t)^* . \out \equiv 0$.
As we shall later see, in our setting, we are interested in identifying \emph{all} (loop-free) counterexamples. Hence, we propose a notion of in-out safe behaviour for which, whenever violated, we can provide all relevant bad behaviours.
\end{remark}

Going back to the example in Section~\ref{sec::prelim}, assume a new programmer $P_2$ which has to enable traffic only from $H_3$ to $H_4$. Assuming the network in Figure~\ref{fig::simpl-net}, $P_2$ encodes the {\PHbH} switch policy:
\begin{equation}\label{eq::pol-p2}
p_2  \triangleq \pt = 3 . \pt \leftarrow 5 + \pt = 6 . \pt \leftarrow 4
\end{equation}
The end-to-end behaviour can be proven correct, by showing that:
\begin{enumerate}
\item packets at port $3$ reach port $4$, i.e.,
\begin{equation}\label{eq::p2-reach}
\vdash (\pt = 3) . (p_2 . t)^* . (\pt = 4) \not \equiv 0
\end{equation}
\item no packets at port $3$ can reach ports $1$ or $2$, i.e.,
\begin{equation}\label{eq::p2-not-reach}
\vdash (\pt = 3) . (p_2 . t)^* . (\pt = 1 + \pt = 2) \equiv 0.
\end{equation}
\end{enumerate}

Nevertheless, it is easy to show that the composed policies $p_1$ in~(\ref{eq::pol-p1}) and $p_2$ in~(\ref{eq::pol-p2}) do not guarantee a safe behaviour. Namely, in the context of the {\PHbH} policy $p_1 + p_2$, packets at port $1$ can reach ports $3$ or $4$, and packets at port $3$ can reach ports $1$ or $2$. This violates the correctness properties in~(\ref{eq::p1-not-reach}) and~(\ref{eq::p2-not-reach}), respectively:
\begin{equation}\label{eq::p12-not-safe1}
\vdash (\pt = 1) . ((p_1 + p_2) . t)^* . (\pt = 3 + \pt = 4)\not \equiv 0
\end{equation}
\begin{equation}\label{eq::p12-not-safe2}
\vdash (\pt = 3) . ((p_1 + p_2) . t)^* . (\pt = 1 + \pt = 2)\not \equiv 0
\end{equation}

In the next section, we would like to provide the explanation for the failure of the network safety, as expressed in~(\ref{eq::p12-not-safe1}) and~(\ref{eq::p12-not-safe2}).

\section{Explaining Safety Failures}\label{sec::expln-fail}

Naturally, the first attempt to explain safety failures is to derive the counterexamples according to the NetKAT axiomatisation. Take, for instance, the end-to-end behaviour $(\pt = 1) . ((p_1 + p_2) . t)^* . (\pt = 3 + \pt = 4)$ in~(\ref{eq::p12-not-safe1}). The axiomatisation leads to the following equivalence:
\begin{equation}\label{eq::simple-counterexample}
(\pt = 1) . ((p_1 + p_2) . t)^* . (\pt = 3 + \pt = 4) \equiv (\pt = 1 . \pt \leftarrow 4) + e
\end{equation}
where $e$ is a NetKAT expression containing the Kleene $*$. A counterexample can be immediately spotted, namely: $\pt = 1 . \pt \leftarrow 4$. Nevertheless, the information it provides is not intuitive enough to serve as an explanation of the failure. Moreover, $e$ can hide additional counterexamples revealed after a certain number of $*$-unfoldings according to KA-UNROLL-R/L in Figure~\ref{fig::axiom-packet}.

In what follows, the focus is on the following two questions:
\begin{enumerate}
\item[$Q_1$:] Can we reveal \emph{more information} within the counterexamples witnessing safety failures?
\item[$Q_2$:] Can we reveal \emph{all the counterexamples} hidden within NetKAT expressions containing $*$?
\end{enumerate}

The answer to $Q_1$ is relatively simple: yes, we can reveal more information on how the packets travel across the topology by inhibiting the PA-MOD-MOD and PA-FILTER-MOD axioms in Figure~\ref{fig::axiom-packet}. Recall that, intuitively, this axiom records only the last modification from a series of modifications of the same field.

The answer to $Q_2$ lies behind the following two observations. (1) From a practical perspective, in order to explain failures it suffices to look at loop-free forwarding paths within the network topology. Reaching the same port twice along a path does not add insightful information about the reason behind the violation of a safety property, as the network behaviour is preserved in the context of that port. Considering loop-free paths is also in accordance with the minimality criterion invoked in the seminal work on causal reasoning in~\cite{DBLP:conf/sum/Halpern11}, for instance.
(2) A {\PHbH} switch policy $p$ of size $n$ entails loop-free paths from $\inn$ to $\out$ crossing at most $n$ switches within a topology $t$. Hence, in order to determine all loop-free paths from $\inn$ to $\out$ it suffices to apply the {\PHbH} policy $p$ along $t$ for $n$ times. Showing that the suggested approximation is sound reduces to showing:
\begin{equation}\label{eq:e2e-n}
\vdash \inn . (p.t)^*. \out \equiv 0 {\textnormal{~~iff~~}} \vdash \inn . (1 + p.t)^n. \out \equiv 0
\end{equation}

Lemma~\ref{lm::expand-m} and Lemma~\ref{lm::expand-m-star} are needed in order to prove the equivalence in~(\ref{eq:e2e-n}).
\begin{lemma}\label{lm::expand-m}
Let $p,\, t$ be two NetKAT policies. The following holds, for all natural numbers $n$:
\begin{equation}\label{eq::unfold}
(1 + p.t)^n \,\, \equiv \,\, 1 + p.t + (p.t)^2 + \ldots + (p.t)^n
\end{equation}
\end{lemma}
\begin{proof}
The proof follows by induction on $n$.\\
\emph{Base case:} $n = 0$. If $n = 0$ then $(1 + (p . t))^0 = 1$, inferred based on the definition of Kleisli composition.\\
\emph{Induction step:} Assume~(\ref{eq::unfold}) holds for all $k$ such that $0 \leq k \leq n$. It follows that:
\[
\begin{array}{rcll}
(1 + p.t)^{n+1} & \equiv & (1 + p.t)^n . (1 + p.t) & \textnormal{(Kleisli composition)}\\
& \equiv & (1 + p.t + (p.t)^2 + \ldots + (p.t)^n).(1 + p.t) & \textnormal{(ind. hypothesis)}\\
& \equiv & 1 + p.t + (p.t)^2 + \ldots + (p.t)^n + &  \\
&    & ~~~~~~ p.t + (p.t)^2 + \ldots + (p.t)^n + (p.t)^{n+1}&  \textnormal{(KA-SEQ-DIST-L/R, KA-ONE-SEQ)}\\
& \equiv & 1 + p.t + (p.t)^2 + \ldots + (p.t)^{n} + (p.t)^{n+1} & \textnormal{(KA-PLUS-IDEM)}
\end{array}
\]
Hence,~(\ref{eq::unfold}) holds.
\end{proof}

\begin{lemma}\label{lm::expand-m-star}
Let $p,\, t,\, \inn,\, \out$ be NetKAT policies. The following holds, for all natural numbers $n$:
\begin{equation}\label{eq::unfold-star}
\inn.(1 + p.t)^n.\out \,\, \leq \,\, \inn.(p.t)^*.\out
\end{equation}
\end{lemma}
\begin{proof}
%
First, observe that
\begin{equation}\label{eq:rew-star-n-plus-one}
\inn.(p.t)^*.\out \equiv \inn.(1 + p.t + (p.t)^2 + \ldots + (p.t)^{n} + (p.t)^{n+1}.(p.t)^* ).\out
\end{equation}
by KA-UNROLL-L/R, KA-PLUS-IDEM and KA-SEQ-DIST-L/R.
Consequently, by Lemma~\ref{lm::expand-m}, the following also holds:
\begin{equation}\label{eq:rew-star-n-plus-one-lm1}
\inn.(p.t)^*.\out \equiv \inn.(1 + p.t)^{n}.\out +  \inn.(p.t)^{n+1}.(p.t)^*.\out
\end{equation}
Therefore,
\[
\inn.(1 + p.t)^{n}.\out \,\, \leq \,\, \inn.(p.t)^*.\out
\]
holds as well.
\end{proof}

\begin{theorem}(Star Elimination for Safety)\label{thm::star-elim-safety}
Assume a network topology $t$, a {\PHbH} switch policy $p$ of size $n$, an ingress policy $\inn \in \Sigma_{\Tst^*}$, and an egress policy $\out \in \Sigma_{\Tst^*}$ encoding the hazard. The following holds:
\begin{equation}\label{eq:star-elim-safe}
\vdash \inn . (p.t)^*. \out \equiv 0 {\textnormal{~~iff~~}} \vdash \inn . (1 + p.t)^n. \out \equiv 0
\end{equation}
\end{theorem}
\begin{proof}
The ``if'' case follows immediately, as by Lemma~\ref{lm::expand-m-star} and the hypothesis the following holds:
\[
0 \leq \inn.(1 + p.t)^{n}.\out \,\, \leq \,\, \inn.(p.t)^*.\out \equiv 0.
\]

For the ``only if'' case we proceed by reductio ad absurdum.

Assume:
\begin{equation}\label{eq::red-abs}
\inn . (p.t)^*. \out \not \equiv 0.
\end{equation}
By the hypothesis, the definition of Kleene $*$ and the assumption in~(\ref{eq::red-abs}), it follows that there exists $N > n$ such that:
\[
\inn . (p.t)^N. \out \not \equiv 0.
\]
Consider:
\begin{equation}\label{eq::explicit-p-t}
\begin{array}{rcl}
p & \triangleq & \Sigma_{i \in \{1, \ldots, n\}} \, \tst^*_i \,.\, \pt \leftarrow v_i\\
t &  \triangleq & \Sigma_{j \in \{1, \ldots, m\}} \pt = v_j \,.\, \pt  \leftarrow v'_j
\end{array}
\end{equation}
and assume $N = n+k$.
Given the shape of the {\PHbH} switch policy $p$ and topology $t$, it follows that there exist policies $e_1 \not \equiv 0$, $e_2$ such that:
\[
\vdash \inn . (p.t)^N. \out  \equiv e_1 + e_2
\]
with $e_1 \triangleq \inn.p_N. \out$ and
\begin{equation}\label{eq::split-pol}
p_N \triangleq (\tst^*_{i_1} \,.\, \pt \leftarrow v_{i_1}\,.\, \pt = v_{i_1} . \pt \leftarrow v'_{j_1}).
\,\, \ldots\,\, . (\tst^*_{i_{n+k}} \,.\, \pt \leftarrow v_{i_{n+k}} \,.\, \pt = v_{i_{n+k}} . \pt \leftarrow v'_{j_{n+k}} )
\end{equation}
where $ \tst^*_{i_l}, v_{i_l}, v'_{j_l}$ range over $ \tst^*_{i}, v_{i}, v'_{j}$ as in~(\ref{eq::explicit-p-t}).

Note that $e_1$ denotes a path that goes through $N > n$ switches from $\inn$ to $\out$.
Recall that the switch policy $p$ can define loop-free paths traversing at most $n$ switches (the size of $p$). Hence, we conclude that the identified path is not loop-free, i.e.:
\begin{equation}\label{eq::e1-loop}
\begin{array}{rcl}
e_1 & \triangleq & \inn.\\
&& (\tst^*_{i_1} \,.\, \pt \leftarrow v_{i_1}\,.\, \pt = v_{i_1} . \pt \leftarrow v'_{j_1}).\\
&& \,\, \ldots\,\, \\
&& (\tst^*_{i_{\alpha}} \,.\, {\bf \pt = v_{i_\alpha}}\,.\,  \pt \leftarrow v'_{i_\alpha}).\\
&& \,\, \ldots\,\, \\
&& (\tst^*_{i_{\alpha + \beta}} \,.\, {\bf \pt = v_{i_\alpha}}\,.\,  \pt \leftarrow v'_{i_{\alpha + \beta}}).\\
&& \,\, \ldots\,\, \\
&& (\tst^*_{i_{n+k}} \,.\, \pt \leftarrow v_{i_{n+k}} \,.\, \pt = v_{i_{n+k}} . \pt \leftarrow v'_{j_{n+k}} ).\\
&&\out
\end{array}
\end{equation}
with $\beta \geq k$.
Based on $e_1$, we can devise a policy defining a loop-free path that crosses at most $n$ switches from $\inn$ to $\out$.
%
Consider a policy $p_N\!\!\downarrow$ that stands for $p_N$ without the loop of size $\beta$ from~(\ref{eq::e1-loop}):
\begin{equation}\label{eq::path-restricted}
\begin{array}{rcl}
p_N\!\!\downarrow & \triangleq  & \tst^*_{i_1} \,.\,\pt \leftarrow v_{i_1}\,.\, \pt = v_{i_1} . \pt \leftarrow v'_{j_1}. \\
&& \ldots\\
&& \tst^*_{i_\alpha} \,.\ \pt = v_{i_\alpha}\,.\,  \pt \leftarrow v'_{i_\alpha}.\\
&& \tst^*_{i_{\alpha + \beta}} \,.\ \pt = v_{i_{\alpha + \beta+1}}\,.\,  \pt \leftarrow v'_{i_{\alpha + \beta+1}}.\\
&& \ldots\\
&& \tst^*_{i_{n+k}} \,.\, \pt \leftarrow v_{i_{n+k}} \,.\, \pt = v_{i_{n+k}} . \pt \leftarrow v'_{j_{n+k}} .\\
\end{array}
\end{equation}

By the construction of $p_N\!\!\downarrow$ and the fact that $e_1 \not \equiv 0$ it follows that:
\begin{equation}\label{eq::contradiction-path}
\inn . p_N\!\!\downarrow . \out \not \equiv 0.
\end{equation}
In words, we identified a path that traverses the topology from $\inn$ to $\out$ and crosses at most $n$ switches.
Moreover, by the hypothesis and Lemma~\ref{lm::expand-m}, the following hold:
\begin{equation}\label{eq::explicit-hyop}
\begin{array}{rcll}
\inn . \out  &  \equiv & 0 & \\
 \inn . (p.t)^i. \out & \equiv & 0 & \textnormal{for all $i \in \{1,\ldots, n\}$}.
\end{array}
\end{equation}
Hence, (\ref{eq::contradiction-path}) contradicts~(\ref{eq::explicit-hyop}).
We conclude that the ``only if'' case holds as well.
\end{proof}

With these ingredients at hand, in accordance with $Q_1$ and $Q_2$, consider an alteration of the axiomatisation as follows. Firstly, we inhibit the axioms PA-MOD-MOD, PA-FILTER-MOD and KA-UNROLL-R/L. Then, we add the star elimination axiom corresponding to~(\ref{eq:star-elim-safe}) in Theorem~\ref{thm::star-elim-safety}:
\begin{equation}\label{eq::star-elim-safe-axiom}
\inn . (p.t)^*. \out \equiv \inn . (1 + p.t)^n. \out .
\end{equation}
Let $\svdash$ be the entailment relation over the modified axiomatisation.

\begin{definition}[Safety Failure Explanations]\label{def::safety-fail-expl}
Assume a network topology $t$, a {\PHbH} switch policy $p$ of size $n$, an ingress policy $\inn \in \Sigma_{\Tst^*}$, and an egress policy $\out \in \Sigma_{\Tst^*}$ encoding the hazard. A \emph{safety failure explanation} is a policy {${\it expl} \not \equiv 0$} in canonical form (i.e., it cannot be reduced further) such that:
\begin{equation}\label{eq:star-elim-safe2}
\svdash \inn . (p.t)^*. \out \equiv {\it expl}.
\end{equation}
\end{definition}

For an example, we refer to the case of the two programmers providing switch policies $p_1$ and $p_2$ forwarding packets from host $H_1$ to $H_2$, and from $H_3$ to $H_4$ within the network in Figure~\ref{fig::simpl-net}. As previously discussed, the end-to-end network behaviour defined over each of the aforementioned policies can be proven correct using the NetKAT axiomatisation. Nevertheless, a comprehensive explanation of what caused the erroneous behaviour over the unified policy $p_1 + p_2$ could not be derived according $\vdash$. The new axiomatisation, however, entails the following explanation:
\[
\svdash (\pt = 1) . ((p_1 + p_2) . t)^* . (\pt = 3 + \pt = 4) \equiv \pt = 1 . \pt \leftarrow 5 . \pt \leftarrow 6 . \pt \leftarrow 4
\]
showing how packets at port $1$ can reach port $4$.
Similarly,
\[
\svdash (\pt = 3) . ((p_1 + p_2) . t)^* . (\pt = 1 + \pt = 2) \equiv \pt = 3 . \pt \leftarrow 5 . \pt \leftarrow 6 . \pt \leftarrow 2
\]
shows how packets at port $3$ can reach port $2$.

\begin{remark}
The work in~\cite{DBLP:conf/popl/AndersonFGJKSW14} proposes a ``star elimination'' method for switch policies not containing $\bf{dup}$ and switch assignments. The procedure in~\cite{DBLP:conf/popl/AndersonFGJKSW14} employs a notion of normal form to which each NetKAT policy can be reduced. The reason for not using the aforementioned star elimination in our context is that the normal forms in~\cite{DBLP:conf/popl/AndersonFGJKSW14} ``forget'' the intermediate sequences of assignments and tests, and reduce policies to sums of expressions of shape $(f_1 = v_1 .\, \ldots \,. f_n = v_n). (f_1 \leftarrow v'_1 .\, \ldots \,. f_n \leftarrow v'_n)$ where $f_1, \ldots, f_n$ are the packet fields. Hence, the normal forms exploited by the star elimination in~\cite{DBLP:conf/popl/AndersonFGJKSW14} can not serve as comprehensive failure explanations.
\end{remark}

\begin{remark}
Note that the safety failure explanations in Definition~\ref{def::safety-fail-expl} are not minimal.
There might be cases in which two explanation paths
\[
\begin{array}{rl}
e_1 \triangleq p'.p'' ~~~~~~~~~~~~ e_2 \triangleq p'.\tilde{p}.p''
\end{array}
\]
are identified.
Nevertheless, minimality in this context is easy to achieve in a post-processing step that pattern-matches expressions like $e_1$ and $e_2$ and keeps only $e_1$ as relevant explanation.
\end{remark}

\section{Discussion}\label{sec::discussion}

In this paper we formulate a notion of safety in the context of NetKAT programs~\cite{DBLP:conf/popl/AndersonFGJKSW14} and provide an equational framework that computes all relevant explanations witnessing a bad, or an unsafe behaviour, whenever the case. The proposed equational framework is a slight modification of the sound and complete axiomatisation of NetKAT, and is parametric on the size of the considered  hop-by-hop switch policy. Our approach is orthogonal to related works which rely on model-checking algorithms for computing all counterexamples witnessing the violation of a certain property, such as~\cite{DBLP:conf/vmcai/Leitner-FischerL13,DBLP:journals/corr/abs-1901-00588}, for instance.
A corresponding tool for automatically computing the explanations can be straightforwardly implemented in a programming language like Maude~\cite{DBLP:conf/rta/ClavelDELMMQ99}, for instance; we leave this exercise as future work.
We consider assessing the complexity of the procedure for real case scenarios, on top of benchmarks as in~\cite{DBLP:conf/nsdi/KazemianVM12,DBLP:journals/jsac/KnightNFBR11}, for instance.

The results in this paper are part of a larger project on (counterfactual) causal reasoning on NetKAT.
In~\cite{Lew73}, Lewis formulates the counterfactual
argument, which defines when an event is considered a cause for some
effect (or hazardous situation) in the following way: a) whenever the event
presumed to be a cause occurs, the effect occurs as well, and b) when the presumed cause does not occur,
the effect will not occur either.
The current result corresponds to item a) in Lewis' definition, as it describes the events that have to happen in order for the hazardous situation to happen as well. The next natural step is to capture the counterfactual test in b). This reduces to tracing back the explanations to the level of the switch policy, and rewrite the latter so that it disables the generation the paths leading to the undesired egress. The generation of a ``correct'' switch policy can be seen as an instance of program repair.

In the future we would be, of course, interested in defining notions of causality (and associated algorithms) with respect to the violation of other relevant properties such as liveness, for instance. We would also like to explain and eventually disable routing loops (i.e., endlessly looping between A and B) from occurring. Or, we would like to identify the cause of packets being not correctly filtered by a certain policy. 

\paragraph{Acknowledgements}{The author is grateful to the reviewers of FROM 2019, for their feedback and observations.
Special thanks are addressed to Tobias Kapp\'e, for his useful comments and insight into the formal foundations of NetKAT.}


\end{document}